\documentclass{article}
\usepackage{xcolor}
\usepackage{spconf,amsmath,graphicx,hyperref}
\usepackage{bm}
\usepackage[style=ieee]{biblatex}
\usepackage{amssymb}
\usepackage{amsthm}
\usepackage{tikz}
\usepackage{caption}
\usepackage{algorithm}       
\usepackage{algpseudocode}  

\newtheorem{theorem}{Theorem}
\newtheorem{lemma}{Lemma}

\definecolor{mygreen}{RGB}{0,0,255}
\definecolor{myblue}{RGB}{0,255,0}
\usepackage{algpseudocode}

\algrenewcommand\alglinenumber[1]{\textbf{#1.}}

\makeatletter
\newcommand{\StateCont}[1]{%
  \Statex\hspace*{\dimexpr\ALG@thistlm\relax}#1%
}
\makeatother

\hypersetup{
    colorlinks=true,
    linkcolor=mygreen,
    citecolor=green,
    urlcolor=myblue
}

\title{Scalable Hessian-free Proximal Conjugate Gradient Method for Nonconvex and Nonsmooth Optimization}
\name{Yiming Zhou and Wei Dai}
\address{Department of Electrical and Electronic Engineering, {Imperial College London}, UK}
\begin{document}
\ninept
\maketitle
\begin{abstract}
This work studies a composite minimization involving a differentiable function $q$ and a nonsmooth function $h$, both may be nonconvex. This problems is ubiquitous in signal processing and machine learning yet remains challenging to solve efficiently, particularly when large-scale instances, poor conditioning, and nonconvexity coincide. To address these challenges, we propose a proximal conjugate gradient method (PCG) that matches the fast convergence of proximal (quasi-)Newton algorithms while reducing computation and memory complexity, and is especially effective for spectrally clustered Hessians. Our key innovation is to form, at each iteration, an approximation to the Newton direction based on CG iterations to build a majorization surrogate. We define this surrogate in a curvature-aware manner and equip it with a CG-derived isotropic weight, guaranteeing majorization of a local second-order model of $q$ along the given direction. To better preserve majorization after the proximal step and enable further approximation refinement, we scale the CG direction by the ratio between the Cauchy step length and a stepsize derived from the largest Ritz value of the CG tridiagonal. All curvature is accessed via Hessian–vector products computed by automatic differentiation, keeping the method Hessian-free. Convergence to first-order critical points is established. Numerical experiments on CS-MRI with nonconvex regularization and on dictionary learning, against benchmark methods, demonstrate the efficiency of the proposed approach.
\end{abstract}

\section{Introduction}

\label{sec:intro}
In this paper, we consider the optimization problem
\begin{equation}
    \min_{\bm{x} \in \mathbb{R}^n} f(\bm{x}) := q(\bm{x}) + h(\bm{x}),
    \label{eq:objective-1}
\end{equation}
where $q:\mathbb{R}^n \to \mathbb{R} \in C^2$ is proper with an $L$-Lipschitz continuous gradient (possibly nonconvex), and $h:\mathbb{R}^n \to \mathbb{R}\cup\{+\infty\}$ is proper and lower semicontinuous (possibly nonconvex and nonsmooth) with an efficiently computable proximal operator, meaning that for any proximal radius $\tau>0$, $\operatorname{prox}_{\tau h}(\bm{y})
:= \arg\min_{\bm{x}\in\mathbb{R}^n}\left\{\, h(\bm{x}) + \frac{1}{2\tau}\|\bm{x}-\bm{y}\|^2 \right\}$ can be evaluated easily. We further assume that $f(\bm{x})$ is  coercive, i.e., $\lim _{\|\bm{x}\|_{2} \rightarrow \infty} f(\bm{x})=+\infty$, to ensure the existence of a minimizer.

The model \eqref{eq:objective-1} is widely applicable in signal and image processing, machine learning, and related fields. Representative instances include the lasso and related formulations, which typically aim to reconstruct blurred or incomplete data or to perform classification \cite{ting2009sparse,zou2006adaptive,abramovich2018high}. In such problems, $q$ denotes a smooth data-fidelity term, for example, a quadratic or logistic loss for given data, and $h$ serves as a sparsity-inducing regularizer. Common choices include the $\ell_{0}$ penalty, SCAD, and MCP \cite{fan2001variable,blumensath2009iterative,zhang2010nearly}. Extending to the low-rank optimization \cite{chi2019nonconvex}, the smooth term typically has the form $q(\bm{X})=\|\mathcal{A}(\bm{X})-\bm{Y}\|^2$ where $\mathcal{A}(\cdot)$ is a linear operator. Different choices of $\mathcal{A}$ correspond to different tasks, e.g., the subsampling operator for low-rank matrix completion and a Hankel lifting for line spectral estimation from incomplete samples \cite{dai2012geometric,zhou2025efficient,yao2025low}. $h$ usually penalizes the singular values of $\bm{X}$ to promote the low-rank property, e.g., $l_{p}$ norm and the indicator of a rank constraint \cite{recht2010guaranteed,tanner2013normalized,li2020rank}. For the nonconvex $q$, one may consider the bilinear form $q(\bm{U},\bm{V}) = \|\mathcal{A}(\bm{U}\bm{V}^{\top})-\bm{Y}\|^2$ and $h$ can be separable $h(\bm{U},\bm{V}) = h_1(\bm{U})+h_2(\bm{V})$ to impose structure on variables. Examples include nonnegative matrix factorization~\cite{lee2000algorithms}, where $\mathcal{A} = \bm{I}$ and $h$ enforces elementwise nonnegativity, and dictionary learning \cite{tovsic2011dictionary} where $h_1$ constrains atoms to unit $\ell_2$ norm and $h_2$ promotes the sparsity. 


Many algorithms have been developed to solve \eqref{eq:objective-1}. A basic approach is proximal gradient (PG) \cite{combettes2011proximal}, which at each iteration performs a gradient step on $q$ followed by a proximal step for $h$. Given the convexity assumption on $f$, PG is shown to converge globally with a sublinear rate  of $\mathcal{O}(1/k
)$ in terms of the function value, where $k$ is the iteration count. However, like other first-order methods, PG is sensitive to ill-conditioning and non-convexity of the problem, often resulting in slow convergence. A standard remedy is Nesterov acceleration, such as FISTA \cite{beck2017first} for convex $f$. Subsequent work \cite{li2015accelerated} extends this approach to nonconvex objectives by reverting to a proximal gradient step whenever the accelerated step is unsatisfactory. Despite these extensions, formal guarantees for accelerated rates $\mathcal{O}(1/k^2
)$ are, to date, primarily available under convexity. 

Beyond first-order methods, proximal Newton and quasi-Newton schemes have been extensively studied. The work in \cite{lee2014proximal} develops a unified proximal Newton framework and establishes convergence results for both exact and inexact subproblem solutions under appropriate Hessian approximation conditions. Their superlinear convergence guarantees often require the strong convexity of $q$. In \cite{yue2019family}, an inexact proximal Newton method with a regularized Hessian is proposed and the assumptions are further relaxed to convex objectives that satisfy the Luo--Tseng error bound property. The later work \cite{kanzow2021globalized} extends these ideas to settings with nonconvex $q$ but convex $h$. These methods typically handle a proximal mapping scaled by the Hessian or its approximation at each iteration, which leads to nonstandard proximal operators. In practice, an inner iterative solver is needed to evaluate the scaled proximal mapping. 

Recent work \cite{zhou2024proximal} avoids the scaled proximal mapping and instead selects a hybrid Newton direction via a tailored majorization principle. This approach requires $q$ to be an exact convex quadratic and the exact inverse of the Hessian. Another class of proximal quasi-Newton methods also directly works with the standard proximal operator \cite{stella2017forward,stella2017simple,themelis2018forward}. These methods operate on the forward–backward envelope (FBE), an exact-penalty reformulation of the composite objective that shares the same set of local minimizers as \eqref{eq:objective-1}. The initial work in \cite{stella2017forward} exploit the continuous differentiability of the FBE for convex $f$ and develop a line-search scheme that minimizes the FBE along limited-memory BFGS (L-BFGS) quasi-Newton directions. Subsequent work \cite{themelis2018forward} develops a nonmonotone line-search proximal quasi-Newton method based on the FBE framework that accommodates nonconvex objectives. They prove global convergence to a critical point when the FBE satisfies the Kurdyka-Łojasiewicz (KŁ) property, and fast local convergence when the Dennis--Mor\'e condition and the strong local optimality hold. As with other quasi-Newton methods, FBE-based algorithms update a Hessian approximation at every iteration. Improving this approximation typically requires storing additional pairs of iterate and gradient differences (e.g., L-BFGS pairs).

In this paper, we develop a proximal conjugate gradient method (PCG) for solving the general nonconvex and nonsmooth problem~\eqref{eq:objective-1}. Our main algorithmic innovation is to find a closed approximation to Newton direction of $q$ from CG iterations to construct a majorization surrogate at each iteration. The surrogate captures curvature along the current CG direction and employs a specially designed isotropic weight that guarantees majorization of a local second-order model of $q$ along that direction. To better preserve majorization after the proximal map of $h$ and permit further CG refinement, we scale the direction by the ratio between the Cauchy step length and a stepsize obtained from the largest Ritz value of the CG tridiagonal. PCG achieves fast convergence comparable to proximal (quasi-)Newton methods while incurring lower computational and memory costs, making it suitable for large-scale problems. Although the proposed method includes an inner loop, when the Hessian spectrum is clustered, CG converges rapidly; moreover, the majorization test often truncates the loop early, keeping the inner-loop cost low. Meanwhile, Hessian--vector products are obtained via automatic differentiation, avoiding forming and storing a large matrix explicitly. Our convergence analysis establishes that every accumulation point of the generated sequence is a first-order critical point. Numerical results on well-known applications verify the effectiveness of the proposed algorithm.

\section{Preliminary}
This section briefly reviews PG and CG used in this paper.

\noindent\textit{Proximal Gradient:} $k$-th iteration of PG solves the isotropic surrogate
\begin{align}
\operatorname*{arg\,min}_{\bm{x}}~q(\bm{x}_k)+\bm{g}_k^{\top}(\bm{x}-\bm{x}_k)+\frac{1}{2\tau}\|\bm{x}-\bm{x}_k\|^2+h(\bm{x}),\label{eq:isotropic-model}
\end{align}
by the standard proximal operator. If $\tau\in(0,1/L)$, PG produces a monotonically nonincreasing objective sequence, and standard arguments show that every cluster point of $\{\bm{x}_k\}$ is a first-order critical point of $f$~\cite{nesterov2013gradient}. PG is often slow on nonconvex and ill-conditioned problems, which has motivated many acceleration schemes; nonetheless, most prove convergence to first-order critical points by appealing to the PG step’s descent property or its optimality conditions. This is because as long as the PG mapping $G_\tau(\bm{x}_k)=\frac{1}{\tau}(\bm{x}_k-\operatorname{prox}_{\tau h}(\bm{x}_k - \tau \bm{g}_k))$ is single-valued then it is the canonical certificate of first-order criticality: $G_\tau\left(\bm{x}^{\star}\right)=0$ iff $0 \in \nabla q\left(\bm{x}^{\star}\right)+\partial h\left(\bm{x}^{\star}\right)$.

\noindent\textit{Conjugate Gradient:} CG is a classical iterative method for solving linear systems.  At an iterate $\bm{x}_k$, denote $\bm{g}_k = \nabla q(\bm{x}_k)$ and $\bm{H}_k = \nabla^2 q(\bm{x}_k)$. To approximately solve $\bm{H}_k \bm{z} = -\bm{g}_k$ for a Newton direction approximation, CG generates $\{\bm{d}_k^{j}\}$ that are $\bm{H}_k$-conjugate ($\bm{d}_k^{j \top}\bm{H}_k\bm{d}_k^{j} = 0$ for $i \neq j$) and directions $\{\bm{z}_k^{j}\}$. Starting from $\bm{z}_k^{0} = 0$, $\bm{r}_k^{0} = \bm{g}_k$, and $\bm{d}_k^{0} = -\bm{g}_k$, the $j$-th loop first checks for negative curvature and terminates if $\bm{d}_k^{j\top} \bm{H}_k \bm{d}_k^j \leq 0$; otherwise the standard CG updates are \cite{nocedal2006numerical}

\begin{align}
\alpha_k^j&=\frac{\bm{r}_k^{j\top} \bm{r}_k^j}{\bm{d}_k^{j\top} \bm{H}_k \bm{d}_k^j}, ~\bm{z}_{k}^{j+1}=\bm{z}_k^j+\alpha_k^j \bm{d}^j_k, ~\bm{r}_{k}^{j+1}=\bm{r}^j_k+\alpha_k^j \bm{H}_k \bm{d}^j_k,\nonumber\\
    \beta_{k}^{j+1}&=\frac{\bm{r}_{k}^{j+1\top} \bm{r}_{k}^{j+1}}{\bm{r}_{k}^{j\top} \bm{r}_{k}^{j}}, ~\bm{d}_{k}^{j+1}=-\bm{r}_{k}^{j+1}+\beta_{k}^{j+1} \bm{d}_k^j. \label{eq:cg-iterates}
\end{align}
\section{Main Results}

In this section, we present how to estimate a step size along the negative gradient direction based on CG coefficients. Next, we introduce a tailored majorization surrogate for generating candidate descent directions and present the overall algorithm. We conclude with a convergence analysis establishing that every accumulation point is first-order critical.

\subsection{Step Size Estimation from CG Coefficients}
The choice of step size along the negative gradient is important for the convergence of proximal algorithms to critical points of \eqref{eq:objective-1}. Many algorithms require a step size $\tau_k \in (0, 1/L)$, where $L$ is the global Lipschitz constant of $\nabla q$ to ensure the sufficient decrease inequality for each iteration \cite{combettes2011proximal,themelis2018forward,liu2024inexact,baraldi2023proximal,zhou2024proximal}

\begin{equation}
    f(\bm{x}_{+}) \leq f(\bm{x}_k) - c_k\|\bm{x}_{+}-\bm{x}_k\|^2,
    \label{eq:sufficient-descent}
\end{equation}
for some $c_k\ge 0$. However, obtaining $L$ entails computing the exact largest eigenvalue $\lambda_{\max}$ of the Hessian, which is costly in large-scale settings. In PCG, we further exploit the CG sequence \eqref{eq:cg-iterates} to estimate $\lambda_{\max}$, thereby \emph{jointly} determining the search direction and the step size. In particular, after $j$ inner CG steps on $\bm{H}_k$, define the $j \times j$ tridiagonal

\begin{equation}
\bm{T}_{k, j}=\left[\begin{array}{cccc}
\delta_{k, 1} & \gamma_{k, 1} & & \\
\gamma_{k, 1} & \delta_{k, 2} & \ddots & \\
& \ddots & \ddots & \gamma_{k, j-1} \\
& & \gamma_{k, j-1} & \delta_{k, j}
\end{array}\right],
\vspace{-0.0em}
\label{eq:tri}
\end{equation}
from the CG coefficients $\left\{\alpha_k^0, \ldots, \alpha_k^{j-1}\right\}$ and $\left\{\beta_k^1, \ldots, \beta_k^{j-1}\right\}$ (set $\beta_k^0:=0$ ) via

\begin{equation}
\begin{aligned}
\delta_{k,1} &:= \frac{1}{\alpha_k^{0}}, ~\delta_{k,\ell} := \frac{1}{\alpha_k^{\ell-1}}+\frac{\beta_k^{\ell-1}}{\alpha_k^{\ell-2}}, ~ \ell=2,\ldots,j,\\
\gamma_{k,\ell} &:= \frac{\sqrt{\beta_k^{\ell}}}{\alpha_k^{\ell-1}},~ \ell=1,\ldots,j-1.
\end{aligned}
\label{eq:lanczos-1}
\end{equation}
Computing the eigenvalues of $\bm{T}_{k, j}$ obtains the Ritz values for $\bm{H}_k$ on the Krylov subspace generated by the CG run, which follows
\begin{equation}
\lambda_{\min }\left(\bm{H}_k\right) \leq \theta_{k, 1}^{(j)} \leq \cdots \leq \theta_{k, j}^{(j)} \leq \lambda_{\max }\left(\bm{H}_k\right) .
\label{eq:lanczos-2}
\end{equation}
At iteration $k$, after each CG step, we compute $\theta_{k,j}^{(j)}$ and set the step size $\tau_k = \delta/|\theta_{k,j}^{(j)}|$ for some $\delta \in (0,1]$. $\delta$ is introduced for analytical convenience. In practice, set $\delta \approx 1$ for strong performance. We repeat this estimation until the sufficient-decrease condition \eqref{eq:sufficient-descent} is satisfied, where

\begin{equation}
    \bm{x}_{+} = \operatorname{prox}_{\tau_k h}\big(\bm{x}_k - \tau_k \bm{g}_k\big).
    \label{eq:suff-de}
\end{equation}
This CG-driven step size search is first applied in proximal algorithms. Unlike geometric backtracking, which repeatedly shrinks from an initial guess, our procedure keeps $\tau_k$ within a data-driven bounded range implied by local spectral estimates, avoiding sensitivity to poor initial choices. Because it leverages local rather than global Lipschitz information, it can let $\tau_k > 1/L$, which accelerates convergence in practice. We demonstrate that this step-size search is well-defined by proving finite termination.

\begin{lemma}
    Fix the iterate $k$ of PCG. Assume no negative curvature is encountered by CG. After $j$ CG steps, form the Lanczos tridiagonal $\bm{T}_{k, j}$ from \eqref{eq:lanczos-1}; let $\theta_{k, j}^{(j)}=\lambda_{\max }\left(\bm{T}_{k, j}\right)$. Set $\tau_k^{(j)} = \delta/|\theta_{k,j}^{(j)}|$ and $\bm{x}_{+}^{(j)}=\operatorname{prox}_{\tau_k^{(j)} h}\left(\bm{x}_k-\tau_k^{(j)} \bm{g}_k\right)$. Then there exists a finite $j$ and a $\delta \in (0,1]$ such that the sufficient decrease condition \eqref{eq:sufficient-descent} holds with $c_k^{(j)}=\frac{1}{2}\left(\frac{|\theta_{k, j}^{(j)}|}{\delta}-L_k\right) \geq 0$, where $L_k = \sup _{t \in[0,1]} \lambda_{\max }\left(\nabla^2 q\left(\bm{x}_k+t \bm{s}\right)\right)$ and $\bm{s} = \bm{x}^{(j)}_+ - \bm{x}_k$.
    \label{lemma:1}
\end{lemma}

\begin{proof}
     We provide the proof sketch due to space limitation. CG on $\bm{H}_k$ is mathematically equivalent to the symmetric Lanczos process starting with $\bm{r}_k/\|\bm{r}_k\|$. Hence $\theta_{k, j}^{(j)}$ increase monotonically and converges to $\lambda_{\max}(\bm{H}_k)$ when $j=n$ \cite{liesen2013krylov}. The $L_k$-smoothness bound for $q$ at $\bm{x}_k$ gives
     
     \begin{equation}
    \begin{aligned}
    f\left(\bm{x}_{+}^{(j)}\right) &\leq f\left(\bm{x}_k\right)-\left(\frac{1}{2 \tau_k^{(j)}}-\frac{L_k}{2}\right)\left\|\bm{x}_{+}^{(j)}-\bm{x}_k\right\|^2\\
    &=f\left(\bm{x}_k\right)-\frac{1}{2}\left(\frac{|\theta_{k, j}^{(j)}|}{\delta}-L_k\right)\left\|\bm{x}_{+}^{(j)}-\bm{x}_k\right\|^2.
    \end{aligned}
    \end{equation}
    Given $\theta_{k, j}^{(j)} \nearrow \lambda_{\max}(\bm{H}_k) \leq  L_k$ and $q \in C^2$, we can have $\frac{|\theta_{k, j}^{(j)}|}{\delta}-L_k>0$ for some finite $j$ and a $\delta \in (0,1]$, whence the claim.
\end{proof}
\noindent CG/Lanczos stepsize estimation is especially effective in several scenarios. When the Hessian spectrum is clustered, the largest Ritz value rapidly approaches $\lambda_{\max}$. For highly indefinite Hessians with $\left|\lambda_{\min }\right| \gg \lambda_{\max }>0$, only a few iterations are needed to obtain $|\theta^{(j)}| \geq \lambda_{\max}$.
\subsection{Descent Direction Selection via Majorization}

The CG directions $\{\bm{z}_k^{j}\}$ capture the local curvature of $q$ at $\bm{x}_k$ and provide promising search directions. Since the composite objective $f$ also involves the (possibly nonsmooth) term $h$, we vet these directions using a majorization surrogate of $f$ and retain those that certify decrease. Recall that $k$-th iteration of proximal Newton-type algorithms typically solves \cite{kanzow2021globalized,liu2024inexact,lee2014proximal}
\begin{equation}
    \arg \min_{\bm{x}} ~\underbrace{q(\bm{x}_k)+\bm{g}_k^{\top}(\bm{x}-\bm{x}_k) + \frac{1}{2} \|\bm{x}-\bm{x}_k\|^2_{\bm{B}_k}}_{m(\bm{x},\bm{x}_k)} + h(\bm{x}),
    \label{eq:local-model-1}
\end{equation}
where $\bm{B}_k$ is either $\bm{H}_k$ or a suitable approximation. Directly handling \eqref{eq:local-model-1} is challenging because the quadratic term $\tfrac{1}{2}\|\bm{x}-\bm{x}_k\|_{\bm{B}_k}^2$ in $m(\bm{x},\bm{x}_k)$ is anisotropic. PG discards curvature by setting $\bm{B}_k=\tfrac{1}{\tau}\bm{I}$ and instead solves the isotropic surrogate \eqref{eq:isotropic-model} but can lead to slow convergence. Drawing on both approaches, we construct a curvature-aware isotropic surrogate and select CG directions along which the surrogate majorizes $q(\bm{x})$, thereby ensuring descent for $f$. In particular, given any $\bm{z}_k\in\{\bm{z}_k^{j}\}$, define 
\begin{equation}
    \Tilde{m}_{\Tilde{\tau}_k}(\bm{x},\bm{x}_k):=q(\bm{x}_k)- \frac{\tau_k}{\tau_k^c}\bm{z}_k^{\top}(\bm{x}-\bm{x}_k)+\frac{1}{2\Tilde{\tau}_k}\|\bm{x}-\bm{x}_k\|^2,
    \label{eq:isotropic-model-2}
\end{equation}
where $\tau_k$ is estimated based on \eqref{eq:suff-de} and $\tau_{k}^c:=\frac{\langle \bm{g}_k,\bm{g}_k\rangle}{\langle \bm{g}_k,\bm{H}_k\bm{g}_k\rangle}$ denotes the Cauchy step length. Two design choices make (2) effective and special. \textbf{First}, we scale the step along $\bm{z}_k$ by $\tau_k/\tau_k^{c}$; by Lemma~\ref{lemma:1}, this ratio is at most $1$. The intuition is that $\tau_k^{c}$ minimize the local model $m(\bm{x},\bm{x}_k)$ along $\bm{g}_k$, while $\tau_k$ enforces sufficient descent for the full objective $f$. Scaling $\bm{z}_k$ by the ratio $\tau_k/\tau_k^{c}$ therefore increases the likelihood that $\bm{z}_k$ serves as a descent direction for $f$. \textbf{Second}, we set the proximal radius $\Tilde{\tau}_k$ according to Lemma~\ref{lemma:proximal-radius}, so that $\Tilde{m}_{\Tilde{\tau}_k}(\bm{x},\bm{x}_k)$ is guaranteed to majorize the quadratic model of $q$ when restricted to the scaled direction.
\begin{lemma}
\label{lemma:proximal-radius}
    Let $\bm{z}_k \neq 0$, and define $A:=\|\bm{z}_k\|^2,~ B:=\bm{z}_k^{\top} \bm{H}_k \bm{z}_k, ~ C:=\langle \bm{g}_k, \bm{z}_k\rangle$. Consider the line $\bm{x} = \bm{x}_k + \alpha \bm{z}_k$ with $\alpha \in \mathbb{R}$. For $\Tilde{\tau}_k>0$, set $a:=\frac{A}{\Tilde{\tau}_k}-B, ~ b:=-(A+C)$. Then, $\forall \alpha \in [0,1]$, the majorization $\tilde{m}_{\Tilde{\tau}_k}\left(\bm{x}_k + \alpha \bm{z}_k,\bm{x}_k\right) \geq m(\bm{x}_k + \alpha \bm{z}_k,\bm{x}_k)$ holds if either $\Tilde{\tau}_k \leq \frac{A}{B}$ and $A+C\leq 0$, or $\frac{A}{A+B+C}\leq \Tilde{\tau}_k \leq \frac{A}{B}$ and $A+C\geq0$.
\end{lemma}

\begin{proof}
    Along $\bm{x} = \bm{x}_k + \alpha \bm{z}_k$, the model difference is $d(\alpha)=\tilde{m}_{\Tilde{\tau}_k}\left(\bm{x}_k + \alpha \bm{z}_k,\bm{x}_k\right) - m(\bm{x}_k + \alpha \bm{z}_k,\bm{x}_k) = -(A+C) \alpha+\frac{1}{2}\left(\frac{A}{\Tilde{\tau}_k}-B\right) \alpha^2 = b \alpha+\frac{1}{2} a \alpha^2$. We discuss different cases that make $d(\alpha) \geq 0$ for all $\alpha \in [0,1]$. If $a \geq 0$ and $b \geq 0$, then $d^{\prime}(\alpha) \geq d^{\prime}(0)=b \geq 0$, so $d$ is nondecreasing on $[0,1]$ and $\min d= d(0)=0$. This yields $\Tilde{\tau}_k \leq \frac{A}{B}$ and $A+C\leq 0$. Let $a \geq 0$ and $b<0$. The stationary point is at $\alpha^*=-\frac{b}{a}>0$. If $0<\alpha^*<1$, then $\min _{[0,1]} d=d\left(\alpha^*\right)=-\frac{b^2}{2a}<0$ which is impossible. Thus we must have $\alpha^* \geq 1$, i.e. $a \leq-b$. Over $[0,1], d$ then decreases, so the minimum is $d(1)=b+\frac{1}{2} a$, which $\geq 0$, i.e. $a \geq-2 b$. Combine everything together we have $\frac{B+2(A+C)}{A} \leq \frac{1}{\Tilde{\tau}_k} \leq \frac{B+(A+C)}{A}$. To retain the same upper bound as the first case, we take $\frac{A}{A+B+C}\le \Tilde{\tau}_k\le \frac{A}{B}$ when $A+C\ge 0$. Note that the case $a \leq 0$ is excluded to make the majorization argument meaningful since we can always find a small $\Tilde{\tau}_k$ to avoid this case.
\end{proof}
\begin{figure}[t]
  \centering
  \includegraphics[width=0.5\linewidth]{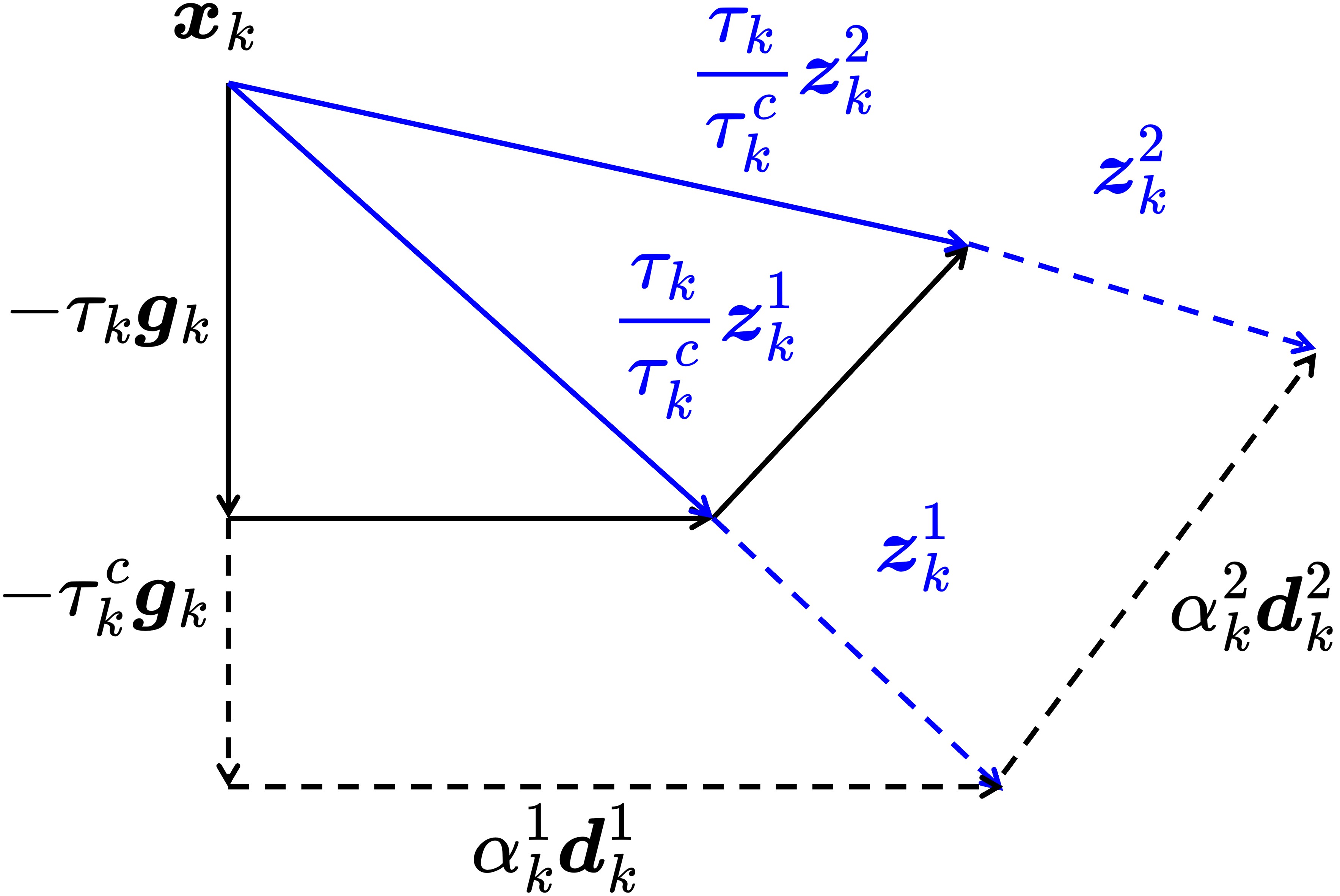} 
  \captionsetup{skip=2pt}
  \caption{Geometric illustration of step scaling in PCG.}
  \label{fig:scaled-directions}
\end{figure}

\noindent Given a $\Tilde{\tau}_k$, we find
\begin{equation}
    \Tilde{\bm{x}} \in \operatorname*{arg\,min}_{\bm{x}}~\Tilde{m}_{\xi\Tilde{\tau}_k}(\bm{x},\bm{x}_k)+h(\bm{x}),
    \label{eq:isotropic-model-3}
\end{equation}
for certain $\xi \in (0,1)$ by the standard proximal operator and certify $\bm{z}_k$ as a descent direction if $q(\Tilde{\bm{x}}) \leq \Tilde{m}_{\Tilde{\tau}_k}(\Tilde{\bm{x}},\bm{x}_k)$. We apply the majorization test along the CG sequence, accepting the last direction that passes and recording the first that fails; the latter typically offers a closer Krylov-subspace approximation to the Newton step. These two directions then serve as endpoints for the subsequent segment backtracking. We summarize the whole algorithm in the following:

\begin{algorithm}
\caption{Proximal Conjugate Gradient Method (PCG)}

  \label{alg:pn-cg-sb-doa}
    \begin{algorithmic}[1]
    \Require Starting point $\bm{x}_k$, and set $\delta \in (0,1]$, $\xi \in (0,1)$.
    \For{$k = 0,1,2,\dots$}
    \State\label{line-2} \hspace{-1em}Update CG iterates based on \eqref{eq:cg-iterates} with negative curvature check \StateCont{\hspace{-1em}and compute $\tau_k$ via \eqref{eq:tri}-\eqref{eq:suff-de} until condition \eqref{eq:sufficient-descent} is met.}
    \State \label{line-3}\hspace{-1em}Identify candidate directions based the majorization principle \StateCont{\hspace{-1em}\eqref{eq:isotropic-model-3}; return $\bm{z}_k^{\text{acc}},\bm{z}_k^{\text{rej}}$, corresponding proximal radius $\Tilde{\tau}_k^{\text{acc}},\Tilde{\tau}_k^{\text{rej}}$,} \StateCont{\hspace{-1em}and $f(\tilde{\bm{x}}_k^{\text{acc}})$.}
    \State\hspace{-1em}Compute $\mu_k^\star := \max\{\mu\in[0,1] \mid 
    f\big(\operatorname{prox}_{\tilde{\tau}_k(\mu) h}\big(\bm{x}_k + \tilde{\bm{z}}_k(\mu)\big)\big) \le f(\tilde{\bm{x}}_k^{\mathrm{acc}})\}$, where $\tilde{\tau}_k(\mu) := \mu \Tilde{\tau}_k^{\text{rej}}+ (1-\mu)\Tilde{\tau}_k^{\text{acc}}$ and $\tilde{\bm{z}}_k(\mu):=\mu \bm{z}_k^{\text{rej}}+ (1-\mu)\bm{z}_k^{\text{acc}}$.
    \State\label{line-4} \hspace{-1em}Update $\bm{x}_{k+1} = \operatorname{prox}_{\tilde{\tau}_k(\mu_k^\star) h}\big(\bm{x}_k + \tilde{\bm{z}}_k(\mu_k^\star)\big)$
    \EndFor
    \end{algorithmic}
\end{algorithm}

Line~\ref{line-3} tests the latest CG direction(s) from Line~\ref{line-2} and continues the CG update if the majorization condition holds. Line~\ref{line-3} is well defined and always identifies a direction that satisfies the majorization condition. In the worst case, setting $\bm{z}_k=-\tau_k^{\mathrm c}\bm{g}_k$ and choosing $\tilde{\tau}_k\le 1$ reduces \eqref{eq:isotropic-model-2} to \eqref{eq:isotropic-model}. Since $\tau_k$ is selected to satisfy \eqref{eq:sufficient-descent}, this case guarantees majorization. PCG is \emph{Hessian-free}: Hessian–vector products are obtained via automatic differentiation without forming $\nabla^2 q$. For any vector $\bm{w}$, $(\nabla q)^{\prime} \bm{w}=\left.\frac{\partial}{\partial \zeta}\left(\left.\nabla q\right|_{\bm{x}+\zeta \bm{w}}\right)\right|_{\zeta=0}.$
In practice, one first evaluates $\nabla q(\bm{x}+\zeta \bm{w})$ and then differentiates with respect to $\zeta$. See \cite[Ch.~8.4]{bright2025matrix} for details. We also provide the global convergence result for Algorithm \ref{alg:pn-cg-sb-doa}.

\begin{theorem}
    Let $\left\{\bm{x}^k\right\}_{k \in \mathbb{N}}$ be a sequence generated by Algorithm \ref{alg:pn-cg-sb-doa}. Then every accumulation point of $\left\{\bm{x}^k\right\}_{k \in \mathbb{N}}$ is a critical point of $f$.
\end{theorem}

\begin{proof}
    Due to space limitations, we sketch the main arguments. We show that the sequence $\{f(\bm{x}_k)\}_{k\in\mathbb{N}}$ is nonincreasing and convergent and $\lim _{k \rightarrow \infty}\left\|\bm{\Delta}_k\right\|^2 \rightarrow 0$, where $\bm{\Delta}_k = \bm{x}_{k+1}-\bm{x}_k$. With these two arguments, adapting the proof template in \cite{li2015accelerated,li2017convergence,frankel2015splitting} then yields the theorem. By the majorization principle and optimality of proximal operator, we have $ q(\bm{x}_{k}) \geq q(\bm{x}_{k+1}) + \frac{\tau_k}{\tau_k^c}\bm{z}_k^{\top}\bm{\Delta}_k-\frac{1}{2\Tilde{\tau}_k}\left\|\bm{\Delta}_k\right\|^2$ and $h(\bm{x}_k) \geq h(\bm{x}_{k+1} )- \frac{\tau_k}{\tau_k^c}\bm{z}_k^{\top}\bm{\Delta}_k+\frac{1}{2\xi\Tilde{\tau}_k}\left\|\bm{\Delta}_k\right\|^2$. Combining two together gets $f\left(\bm{x}_{k+1}\right) \leq f\left(\bm{x}_k\right)-\left(\frac{1}{2\xi\Tilde{\tau}_k}-\frac{1}{2\Tilde{\tau}_k}\right)\left\|\bm{\Delta}_k\right\|^2$ which implies $f(\bm{x}_{k+1}) \leq f(\bm{x}_k)$. Since $f$ is coercive and lsc thus lower bounded, we can conclude the $\{f(\bm{x}_k)\}_{k\in\mathbb{N}}$ is convergent and the sequence $\{\bm{x}_k\}_{k\in\mathbb{N}}$ is bounded. Denote $\bm{x}^*$ and $f^*$ as the accumulation point of $\{\bm{x}_k\}_{k\in\mathbb{N}}$ and the corresponding function value, respectively. Summing over $k$ gives $\sum_{k=0}^{\infty}\Big(\tfrac{1}{2\xi \tilde\tau_k}-\tfrac{1}{2\tilde\tau_k}\Big)\|\bm{\Delta}_k\|^2\le f(\bm{x}_0)-f^\star<\infty$. Because $\{\tilde\tau_k\}$ is bounded and $\xi\in(0,1)$, the weights are bounded below by a positive constant; hence $\|\bm{\Delta}_k\|^2\to 0$ as $k\to\infty$.
\end{proof}

\section{Numerical Results}
We compare PCG with PG \cite{combettes2011proximal}, its extrapolated variant APG \cite{li2015accelerated}, PANOC \cite{stella2017simple}, and ZeroFPR \cite{themelis2018forward}. These methods are chosen because they avoid explicit matrix inversion, which is prohibitive at our problem scales. All experiments were conducted in Julia on a Windows 11 laptop with an Intel Core i7-11800H CPU and 32\,GB of RAM.

\begin{figure}[htbp]
\begin{minipage}[b]{.485\linewidth}
  \centering
\centerline{\includegraphics[width=4.9cm]{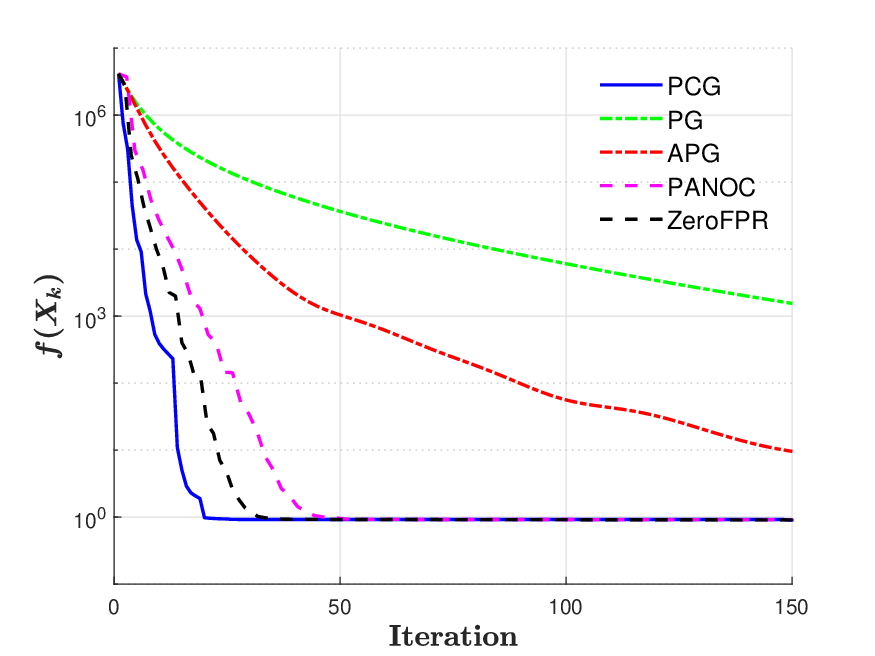}}
\end{minipage}
\hfill
\begin{minipage}[b]{0.485\linewidth}
  \centering
\centerline{\includegraphics[width=4.9cm]{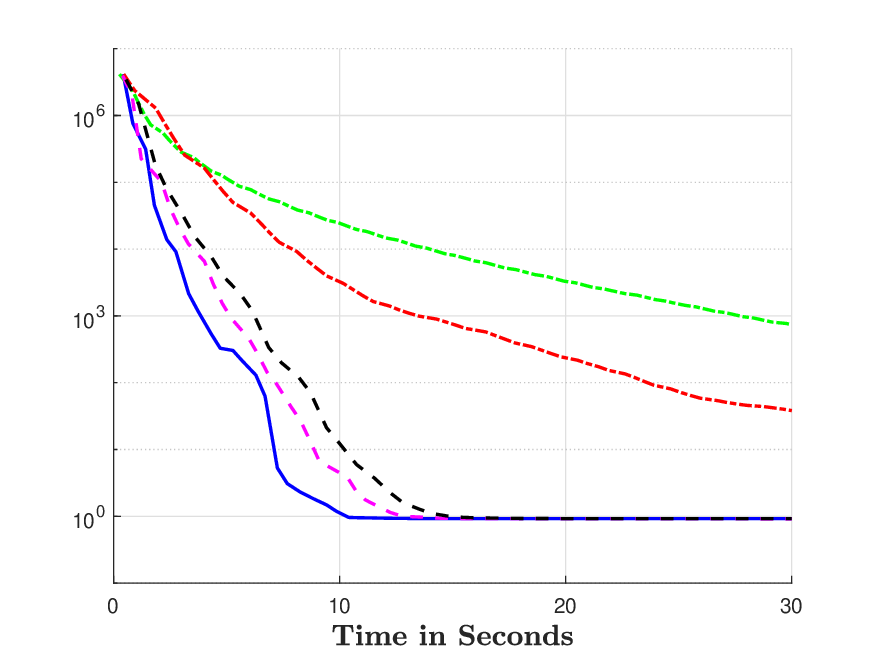}}
\end{minipage}
\begin{minipage}[b]{.485\linewidth}
  \centering
\centerline{\includegraphics[width=4.9cm]{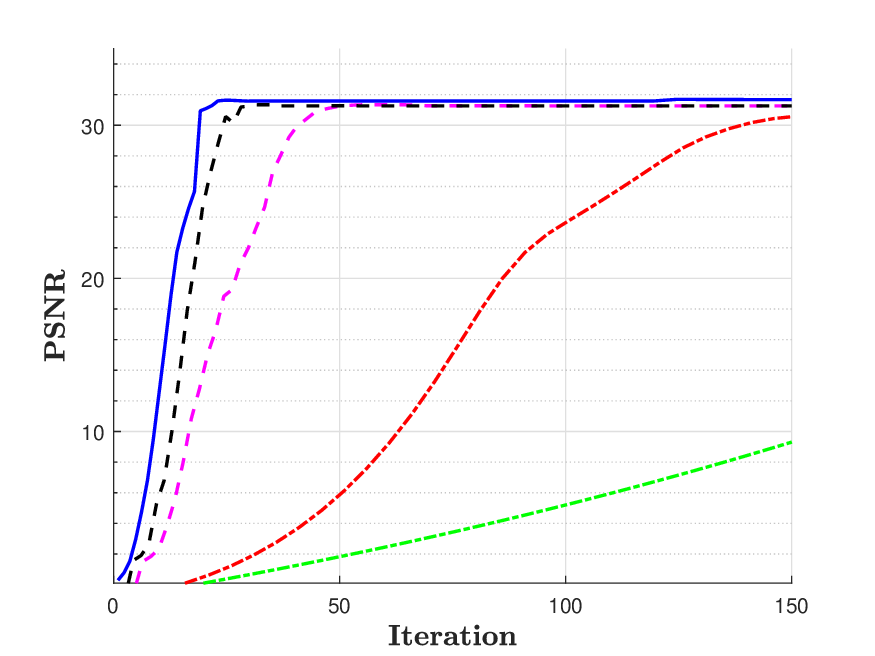}}
\end{minipage}
\hfill
\begin{minipage}[b]{0.485\linewidth}
  \centering
\centerline{\includegraphics[width=4.9cm]{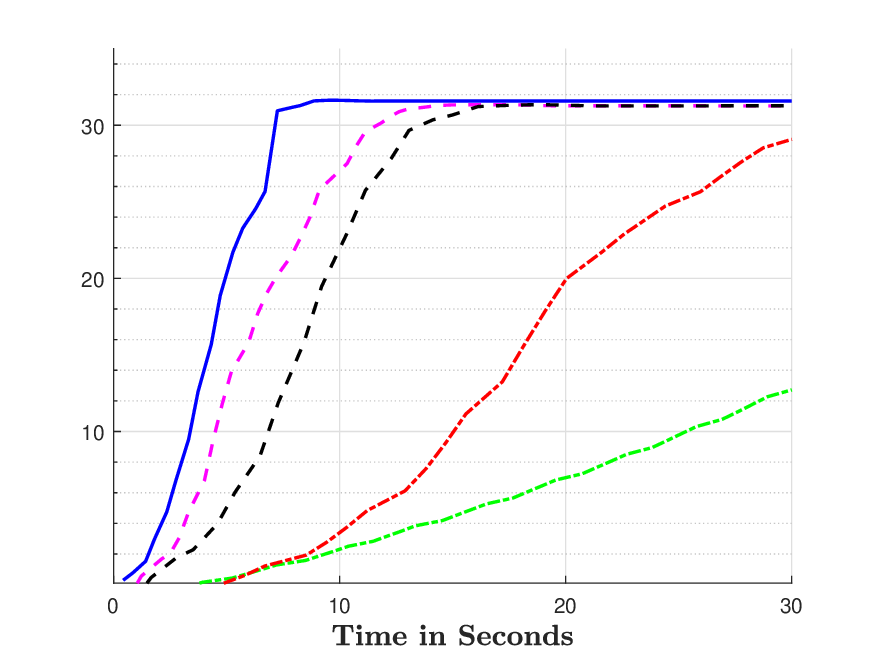}}
\end{minipage}

\caption{ \textit{Top:} Objective value versus iterations and running time. \textit{Bottom:} PSNR comparison, where $\text{PSNR} = 10 \log _{10}\left(\frac{I_{\max }^2}{\mathrm{MSE}}\right) \mathrm{dB}$, $I_{\max }$ denotes the peak pixel value and $\mathrm{MSE}=\frac{\|\bm{X}_{\text{ground-true}}-{\bm{X}_{\text{recover}}}\|_F^2}{512^2}$.}
\label{fig:2}

\end{figure}
\textbf{Compressed sensing MRI:} CS-MRI reconstructs images from undersampled k-space by exploiting wavelet transform sparsity. We use the SCAD penalty $p_{\lambda,a}(\cdot)$ as the sparsity regularizer (definition in \cite{fan2001variable}). The optimization formulation is
\begin{equation}
\min_{\bm{X}} ~\frac{1}{2}\left\|\bm{M} \odot \bm{\mathcal{F}}_{2 D}(\bm{X})-\bm{Y}\right\|_F^2+\sum_i \rho_{\lambda, a}\left(\left|\left[\bm{\mathcal{W}}_{2 D}(\bm{X})\right]_i\right|\right),\nonumber
\end{equation}
where $\bm{M}$ is the sampling mask, $\bm{\mathcal{F}}_{2 D}$ denotes the 2-D FFT, $\bm{Y}$ is the measured k-space data, and $\bm{\mathcal{W}}_{2 D}$ is the 2-D wavelet transform. Since $\bm{\mathcal{W}}_{2 D}^*\bm{\mathcal{W}}_{2 D} = \bm{\mathcal{W}}_{2 D}\bm{\mathcal{W}}_{2 D}^*=\bm{I}$ denotes an orthonormal transform, the proximal operator of $h$ can be computed as $\operatorname{prox}_{\tau h}(x)=\bm{\mathcal{W}}_{2 D}^* \operatorname{prox}_{\tau p_{\lambda,a}(\cdot)}(\bm{\mathcal{W}}_{2 D}(\bm{X}))$. Following \cite{breheny2011coordinate}, for each component $x_i$,

\begin{equation}
\operatorname{prox}_{\tau p_{\lambda, a}}\left(x_i\right)= \begin{cases}\operatorname{sgn}\left(x_i\right) \max \left(0,\left|x_i\right|-\lambda\right), & \left|x_i\right| \leq b_1 \\ \frac{(a-1) x_i-\operatorname{sgn}\left(x_i\right) a \tau \lambda}{a-1-\tau}, & b_1<\left|x_i\right| \leq b_2 \\ x_i, & \left|x_i\right|>b_2,\end{cases}
\nonumber
\end{equation}
with thresholds $b_1 = \frac{\lambda(a-1-\tau+a \tau)}{a-1}$ and $b_2 = a \lambda$. Because the Hessian of $q$ is an orthogonal projector with spectrum $\{1,0\}$, the inner CG loop converges rapidly. Matrix-vector products reduce to (2-D) FFTs, further lowering cost. We evaluate all methods on a $512\times 512$ grayscale brain MRI image~\cite{imagej_mri_stack}. Cartesian $k$-space is undersampled using a variable-density Bernoulli pattern with a fully sampled central disk ($30\%$ of the radius). Outside the center, samples are drawn independently with probability $1/4$. Gaussian noise is added to achieve $\mathrm{SNR}=25\mathrm{dB}$. The regularization parameter was set as $\lambda = 0.002$ and $a =3.7$. The simulation results are
shown in Figure \ref{fig:2}. Our method converges faster than competing methods in both iteration count and computing time. With the same initialization, it also attains a slightly higher PSNR.

\begin{figure}[t]
\begin{minipage}[b]{.490\linewidth}
  \centering
\centerline{\includegraphics[width=4.8cm]{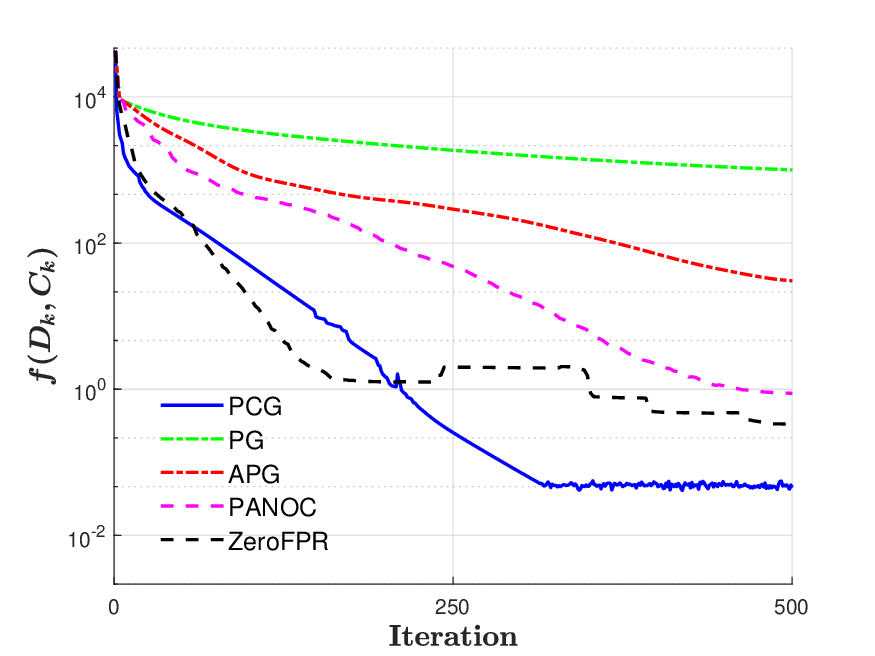}}
\end{minipage}
\hfill
\begin{minipage}[b]{0.490\linewidth}
  \centering
\centerline{\includegraphics[width=4.8cm]{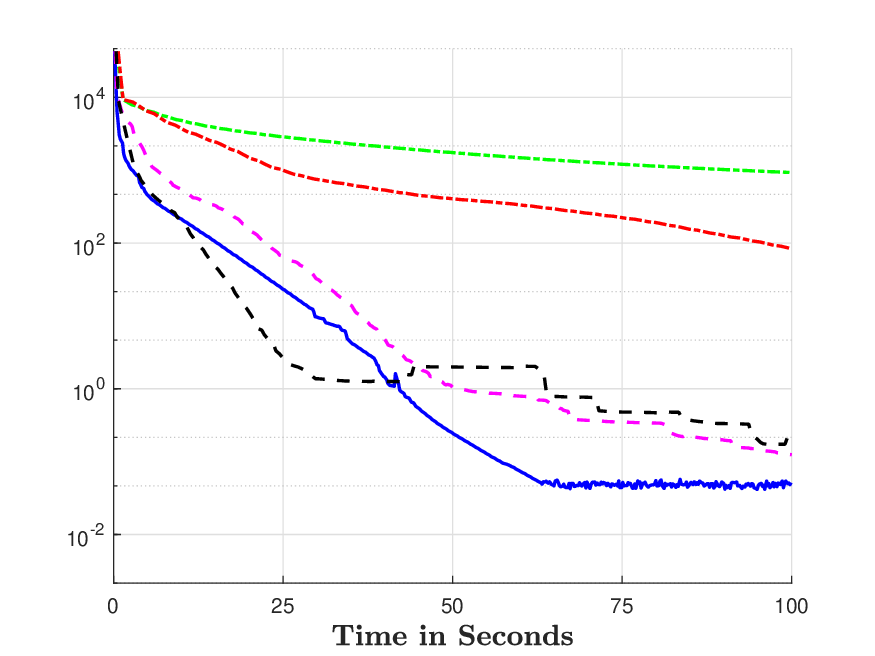}}
\end{minipage}
\begin{minipage}[b]{.490\linewidth}
  \centering
\centerline{\includegraphics[width=4.8cm]{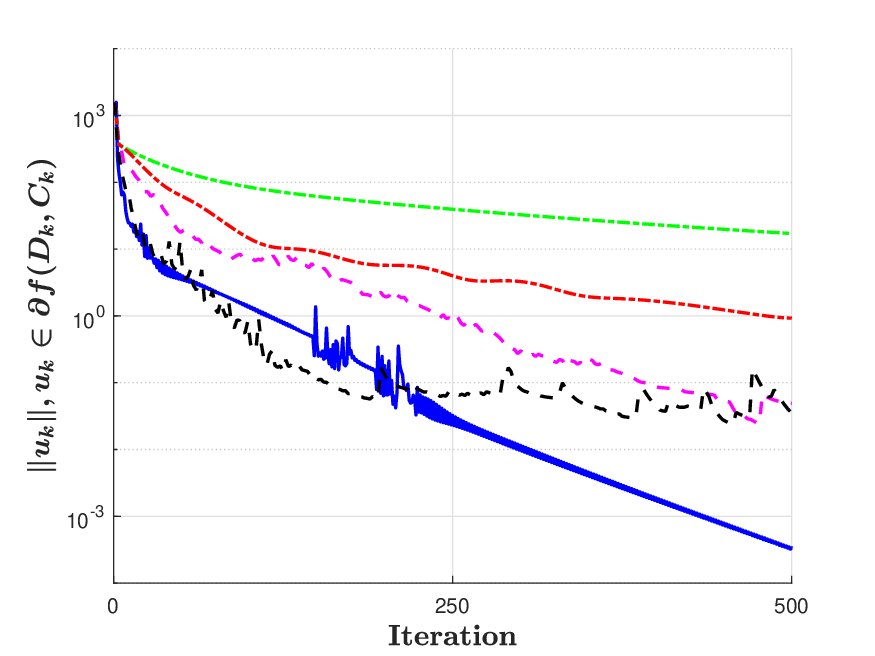}}
\end{minipage}
\hfill
\begin{minipage}[b]{0.490\linewidth}
  \centering
\centerline{\includegraphics[width=4.8cm]{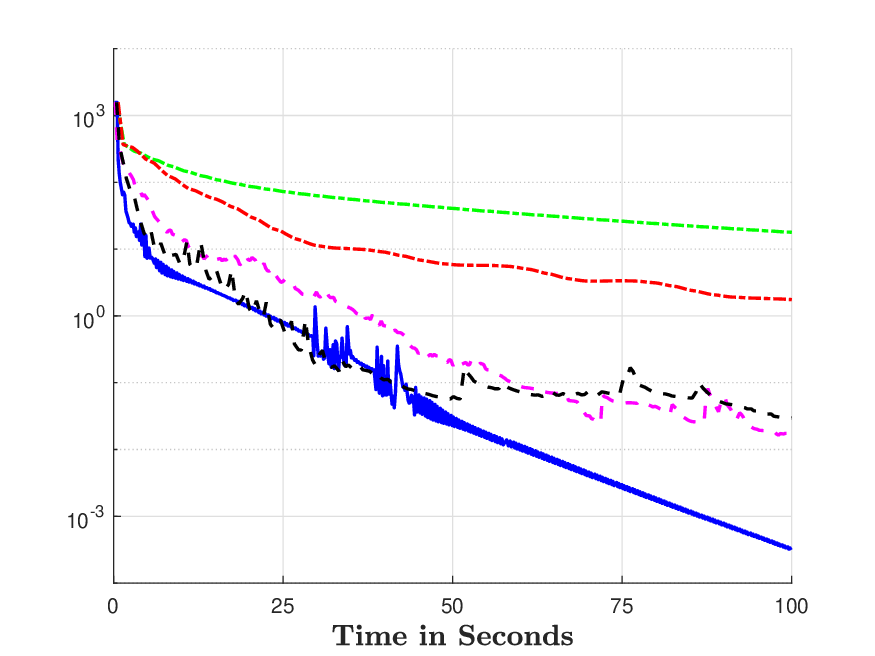}}
\end{minipage}

\caption{\textit{Top:} Objective value versus iterations and running time. \textit{Bottom:} Norm of sub-gradient comparison, where sub-gradient follows \cite{beck2017first}[Definition 10.5]. }
\label{fig:3}
\end{figure}

\textbf{Dictionary Learning:} We next consider a bilinear inverse problem, where the data-fidelity term $q$ is nonconvex. A prototypical case is dictionary learning \cite{aharon2006k,dai2012simultaneous}:
\begin{equation}
    \min_{\bm{D},\bm{C}}~\frac{1}{2}\left\|\bm{D}\bm{C} -\bm{Y}\right\|_F^2 + \sum_{i=1}^r \iota_{\left\{\left\|\bm{d}_i\right\|_2=1\right\}}\left(\bm{d}_i\right)+\sum_{j=1}^n \iota_{\left\{\left\|\bm{c}_j\right\|_0 \leq k\right\}}\left(\bm{c}_j\right),\nonumber
\end{equation}
where $\bm D=[\bm d_1,\ldots,\bm d_r]$, $\bm C=[\bm c_1,\ldots,\bm c_n]$, and $\bm Y\in\mathbb R^{m\times n}$.
Here $\iota_{\mathcal S}(\cdot)$ denotes the indicator of a set $\mathcal S$, i.e.,
$\iota_{\mathcal S}(x)=0$ if $x\in\mathcal S$ and $+\infty$ otherwise.
Thus $\iota_{\{\|\bm d_i\|_2=1\}}(\bm d_i)$ enforces unit-norm atoms and
$\iota_{\{\|\bm c_j\|_0\le k\}}(\bm c_j)$ enforces at-most-$k$ sparsity per code column. Its proximal operator decouples columnwise: it projects $\bm{D}$ by normalizing each atom to unit $\ell_2$ norm and projects $\bm{C}$ by keeping, in each column, the $k$ largest-magnitude entries and zeroing the rest. $\bm{D}_{true}$ is randomly generated Gaussian matrix with unit $\ell_2$ norm columns. Each column of $\bm{C}_{true}$ contains $k$ non-zeros with uniformly random support and i.i.d. standard-Gaussian values. We consider $m=250$, $r=500$, $n=1000$, and $k=10$. Figure \ref{fig:3} illustrates that our method achieves lower reconstruction error with fewer iterations and reduced runtime. It also avoids the late-iteration variable oscillations observed in competing methods.

\section{Conclusions}

This paper develops a proximal CG method for nonconvex nonsmooth optimization. Leveraging CG iterates, we estimate a stepsize via the Lanczos process and choose a descent direction via a tailored majorization strategy. Numerical results on image processing and bilinear inverse problems show fast convergence, low computational cost, and strong reconstruction quality.

\printbibliography

@article{zou2006adaptive,
  title={The adaptive lasso and its oracle properties},
  author={Zou, Hui},
  journal={Journal of the American statistical association},
  volume={101},
  number={476},
  pages={1418--1429},
  year={2006},
  publisher={Taylor \& Francis}
}

@article{abramovich2018high,
  title={High-dimensional classification by sparse logistic regression},
  author={Abramovich, Felix and Grinshtein, Vadim},
  journal={IEEE Transactions on Information Theory},
  volume={65},
  number={5},
  pages={3068--3079},
  year={2018},
  publisher={IEEE}
}

@article{fan2001variable,
  title={Variable selection via nonconcave penalized likelihood and its oracle properties},
  author={Fan, Jianqing and Li, Runze},
  journal={Journal of the American statistical Association},
  volume={96},
  number={456},
  pages={1348--1360},
  year={2001},
  publisher={Taylor \& Francis}
}

@article{blumensath2009iterative,
  title={Iterative hard thresholding for compressed sensing},
  author={Blumensath, Thomas and Davies, Mike E},
  journal={Applied and computational harmonic analysis},
  volume={27},
  number={3},
  pages={265--274},
  year={2009},
  publisher={Elsevier}
}

@article{zhang2010nearly,
  title={Nearly unbiased variable selection under minimax concave penalty},
  author={Zhang, Cun-Hui},
  year={2010}
}

@article{ting2009sparse,
  title={Sparse image reconstruction for molecular imaging},
  author={Ting, Michael and Raich, Raviv and Hero, Alfred O},
  journal={IEEE Transactions on Image Processing},
  volume={18},
  number={6},
  pages={1215--1227},
  year={2009},
  publisher={IEEE}
}

@article{chi2019nonconvex,
  title={Nonconvex optimization meets low-rank matrix factorization: {A}n overview},
  author={Chi, Yuejie and Lu, Yue M and Chen, Yuxin},
  journal={IEEE Transactions on Signal Processing},
  volume={67},
  number={20},
  pages={5239--5269},
  year={2019},
  publisher={IEEE}
}

@article{yao2025low,
  title={A Low-rank Projected Proximal Gradient Method for Spectral Compressed Sensing},
  author={Yao, Xi and Dai, Wei},
  journal={IEEE Transactions on Signal Processing},
  year={2025},
  publisher={IEEE}
}

@inproceedings{zhou2025efficient,
  title={Efficient Gridless Wideband Direction-of-Arrival Estimation From Many Frequencies},
  author={Zhou, Yiming and Fu, Huayu and Dai, Wei},
  booktitle={ICASSP 2025-2025 IEEE International Conference on Acoustics, Speech and Signal Processing (ICASSP)},
  pages={1--5},
  year={2025},
  organization={IEEE}
}

@article{dai2012geometric,
  title={A geometric approach to low-rank matrix completion},
  author={Dai, Wei and Kerman, Ely and Milenkovic, Olgica},
  journal={IEEE Transactions on Information Theory},
  volume={58},
  number={1},
  pages={237--247},
  year={2012},
  publisher={IEEE}
}

@article{li2020rank,
  title={Rank-one matrix approximation with {l}p-norm for image inpainting},
  author={Li, Xiao Peng and Liu, Qi and So, Hing Cheung},
  journal={IEEE Signal Processing Letters},
  volume={27},
  pages={680--684},
  year={2020},
  publisher={IEEE}
}

@article{tanner2013normalized,
  title={Normalized iterative hard thresholding for matrix completion},
  author={Tanner, Jared and Wei, Ke},
  journal={SIAM Journal on Scientific Computing},
  volume={35},
  number={5},
  pages={S104--S125},
  year={2013},
  publisher={SIAM}
}

@article{recht2010guaranteed,
  title={Guaranteed minimum-rank solutions of linear matrix equations via nuclear norm minimization},
  author={Recht, Benjamin and Fazel, Maryam and Parrilo, Pablo A},
  journal={SIAM review},
  volume={52},
  number={3},
  pages={471--501},
  year={2010},
  publisher={SIAM}
}

@incollection{combettes2011proximal,
  title={Proximal splitting methods in signal processing},
  author={Combettes, Patrick L and Pesquet, Jean-Christophe},
  booktitle={Fixed-point algorithms for inverse problems in science and engineering},
  pages={185--212},
  year={2011},
  publisher={Springer}
}

@book{beck2017first,
  title={First-order methods in optimization},
  author={Beck, Amir},
  year={2017},
  publisher={SIAM}
}

@article{li2015accelerated,
  title={Accelerated proximal gradient methods for nonconvex programming},
  author={Li, Huan and Lin, Zhouchen},
  journal={Advances in neural information processing systems},
  volume={28},
  year={2015}
}

@article{lee2014proximal,
  title={Proximal {N}ewton-type methods for minimizing composite functions},
  author={Lee, Jason D and Sun, Yuekai and Saunders, Michael A},
  journal={SIAM Journal on Optimization},
  volume={24},
  number={3},
  pages={1420--1443},
  year={2014},
  publisher={SIAM}
}

@article{yue2019family,
  title={A family of inexact {SQA} methods for non-smooth convex minimization with provable convergence guarantees based on the {L}uo--{T}seng error bound property},
  author={Yue, Man-Chung and Zhou, Zirui and So, Anthony Man-Cho},
  journal={Mathematical Programming},
  volume={174},
  number={1},
  pages={327--358},
  year={2019},
  publisher={Springer}
}

@article{kanzow2021globalized,
  title={Globalized inexact proximal {N}ewton-type methods for nonconvex composite functions},
  author={Kanzow, Christian and Lechner, Theresa},
  journal={Computational Optimization and Applications},
  volume={78},
  number={2},
  pages={377--410},
  year={2021},
  publisher={Springer}
}

@article{themelis2018forward,
  title={Forward-backward envelope for the sum of two nonconvex functions: Further properties and nonmonotone linesearch algorithms},
  author={Themelis, Andreas and Stella, Lorenzo and Patrinos, Panagiotis},
  journal={SIAM Journal on Optimization},
  volume={28},
  number={3},
  pages={2274--2303},
  year={2018},
  publisher={SIAM}
}

@article{stella2017forward,
  title={Forward--backward quasi-{N}ewton methods for nonsmooth optimization problems},
  author={Stella, Lorenzo and Themelis, Andreas and Patrinos, Panagiotis},
  journal={Computational Optimization and Applications},
  volume={67},
  number={3},
  pages={443--487},
  year={2017},
  publisher={Springer}
}

@inproceedings{stella2017simple,
  title={A simple and efficient algorithm for nonlinear model predictive control},
  author={Stella, Lorenzo and Themelis, Andreas and Sopasakis, Pantelis and Patrinos, Panagiotis},
  booktitle={2017 IEEE 56th Annual Conference on Decision and Control (CDC)},
  pages={1939--1944},
  year={2017},
  organization={IEEE}
}

@article{zhou2024proximal,
  title={Proximal dogleg opportunistic majorization for nonconvex and nonsmooth optimization},
  author={Zhou, Yiming and Dai, Wei},
  journal={arXiv preprint arXiv:2402.19176},
  year={2024}
}

@book{nocedal2006numerical,
  title={Numerical optimization},
  author={Nocedal, Jorge and Wright, Stephen J},
  year={2006},
  publisher={Springer}
}

@article{liu2024inexact,
  title={An inexact regularized proximal {N}ewton method for nonconvex and nonsmooth optimization},
  author={Liu, Ruyu and Pan, Shaohua and Wu, Yuqia and Yang, Xiaoqi},
  journal={Computational Optimization and Applications},
  volume={88},
  number={2},
  pages={603--641},
  year={2024},
  publisher={Springer}
}

@article{baraldi2023proximal,
  title={A proximal trust-region method for nonsmooth optimization with inexact function and gradient evaluations},
  author={Baraldi, Robert J and Kouri, Drew P},
  journal={Mathematical Programming},
  volume={201},
  number={1},
  pages={559--598},
  year={2023},
  publisher={Springer}
}

@misc{imagej_mri_stack,
  title        = {{MRI} {S}tack (sample image)},
  author       = {{ImageJ, National Institutes of Health}},
  howpublished = {\url{https://imagej.net/ij/images/mri-stack.zip}},
  note         = {Public domain},
  year         = {n.d.},
  urldate      = {2025-09-08}
}

@book{liesen2013krylov,
  title={Krylov subspace methods: principles and analysis},
  author={Liesen, J{\"o}rg and Strakos, Zdenek},
  year={2013},
  publisher={Numerical Mathematics and Scientific Computing}
}

@article{aharon2006k,
  title={{K-SVD}: {A}n algorithm for designing overcomplete dictionaries for sparse representation},
  author={Aharon, Michal and Elad, Michael and Bruckstein, Alfred},
  journal={IEEE Transactions on signal processing},
  volume={54},
  number={11},
  pages={4311--4322},
  year={2006},
  publisher={IEEE}
}

@article{dai2012simultaneous,
  title={Simultaneous codeword optimization {(SimCO)} for dictionary update and learning},
  author={Dai, Wei and Xu, Tao and Wang, Wenwu},
  journal={IEEE Transactions on Signal Processing},
  volume={60},
  number={12},
  pages={6340--6353},
  year={2012},
  publisher={IEEE}
}

@article{breheny2011coordinate,
  title={Coordinate descent algorithms for nonconvex penalized regression, with applications to biological feature selection},
  author={Breheny, Patrick and Huang, Jian},
  journal={The annals of applied statistics},
  volume={5},
  number={1},
  pages={232},
  year={2011}
}

@inproceedings{li2017convergence,
  title={Convergence analysis of proximal gradient with momentum for nonconvex optimization},
  author={Li, Qunwei and Zhou, Yi and Liang, Yingbin and Varshney, Pramod K},
  booktitle={International Conference on Machine Learning},
  pages={2111--2119},
  year={2017},
  organization={PMLR}
}

@article{frankel2015splitting,
  title={Splitting methods with variable metric for {K}urdyka--{\L}ojasiewicz functions and general convergence rates},
  author={Frankel, Pierre and Garrigos, Guillaume and Peypouquet, Juan},
  journal={Journal of Optimization Theory and Applications},
  volume={165},
  number={3},
  pages={874--900},
  year={2015},
  publisher={Springer}
}

@article{lee2000algorithms,
  title={Algorithms for non-negative matrix factorization},
  author={Lee, Daniel and Seung, H Sebastian},
  journal={Advances in neural information processing systems},
  volume={13},
  year={2000}
}

@article{tovsic2011dictionary,
  title={Dictionary learning},
  author={To{\v{s}}i{\'c}, Ivana and Frossard, Pascal},
  journal={IEEE Signal Processing Magazine},
  volume={28},
  number={2},
  pages={27--38},
  year={2011},
  publisher={IEEE}
}

@article{nesterov2013gradient,
  title={Gradient methods for minimizing composite functions},
  author={Nesterov, Yu},
  journal={Mathematical programming},
  volume={140},
  number={1},
  pages={125--161},
  year={2013},
  publisher={Springer}
}

@article{bright2025matrix,
  title={Matrix Calculus (for Machine Learning and Beyond)},
  author={Bright, Paige and Edelman, Alan and Johnson, Steven G},
  journal={arXiv preprint arXiv:2501.14787},
  year={2025}
}

\end{document}